\documentclass[conference]{IEEEtran}
\IEEEoverridecommandlockouts
\usepackage{textcomp}
\def\BibTeX{{\rm B\kern-.05em{\sc i\kern-.025em b}\kern-.08em
		T\kern-.1667em\lower.7ex\hbox{E}\kern-.125emX}}

\usepackage{color}
\usepackage{amsfonts,amsmath,amssymb,amsthm}
\allowdisplaybreaks[2]
\usepackage{latexsym,amscd}
\usepackage{amsbsy}
\usepackage{graphicx,multirow,bm}
\usepackage{algorithm}
\usepackage{algorithmicx}
\usepackage{algpseudocode}
\usepackage{xcolor}
\usepackage{tikz}
\usepackage{diagbox}
\usepackage{multicol}
\usepackage{arydshln}
\usepackage{bbm}
\usepackage{caption}
\usepackage{times}
\usepackage{amsmath}
\usepackage{cases}
\usepackage{subfigure}
\usepackage{cite}
\usepackage{setspace}
\usepackage[left=0.66in, right=0.67in, top=0.76in, bottom=1.15in]{geometry}
\makeatletter

\newcommand{\mb}{\mathbf}

\newtheorem{theorem}{Theorem}
\newtheorem{lemma}{Lemma}

\newtheorem{definition}{Definition}

\newtheorem{remark}{Remark}

\newtheorem{construction}{Construction}

\newtheorem{observation}{Observation}

\begin{document}

	\title{One Burst of $t$-Deletion and One Burst of $t$-Substitution Error-Correcting Codes}
\author{Yajuan Liu,~\IEEEmembership{Member,~IEEE}, Han Cai,~\IEEEmembership{Member,~IEEE}, and Tolga M. Duman, \IEEEmembership{Fellow, IEEE}
	\thanks{
		Y. Liu, and M. Duman are with the Electrical and Electronics Engineering Department, Bilkent University, Ankara, Turkey (email: yajuan.liu@bilkent.edu.tr, duman@ee.bilkent.edu.tr), and H. Cai is with the Information Security and National Computing Grid Laboratory, Southwest Jiaotong University, Chengdu, China (email: hancai@swjtu.edu.cn).
	}	
	\thanks{This work was funded by the European Union through the ERC Advanced
		Grant 101054904: TRANCIDS. Views and opinions expressed are, however,
		those of the authors only and do not necessarily reflect those of the European
		Union or the European Research Council Executive Agency. Neither the
		European Union nor the granting authority can be held responsible for them.}
	%\thanks{This work was supported in part by the National Natural Science Foundation of China under Grant 62271421.}
}

\maketitle
	
		\begin{abstract}
			Synchronization errors, including insertions, deletions, and substitutions, may occur in bursts
		in communication systems such as DNA data storage, file synchronization, and magnetic recording.
		In this paper, we study an error model consisting of one burst of $t$-deletions and one burst of $t$-substitutions. By reformulating the original sequence into a matrix form, we propose an explicit construction of error-correcting codes capable of correcting one burst of $t$-deletions and one burst of $t$-substitutions with $O(\log n)$ redundancy.
	\end{abstract}

\section{Introduction}	
With the advancement of document synchronization applications \cite{cheng2019block}, magnetic recording channels \cite{levenshtein1993perfect}, and DNA data storage systems \cite{heckel2019acharacterization}, synchronization channels involving insertions, deletions, and substitutions have garnered significant attention in recent years \cite{hu2010achievable,brakensiek2018efficient,sima2021on,guruswami2021explicit}.
In many practical communication and storage scenarios, errors tend to occur in bursts \cite{parkin2008magnetic,mazumdar2011channels,zhang2015hifi}.
% Specifically, insertion, deletion and substitution errors may arise at consecutive positions with a bounded maximum length, forming what are commonly referred to as burst errors.
%The presence of burst synchronization errors poses significant challenges for error correction as such errors typically affect long contiguous substring of the transmitted sequences.
%Consequently, conventional error-correcting codes (ECCs) designed for random errors may exhibit poor performance or incur prohibitively large redundancy when applied to burst-error settings.
Therefore, it is of fundamental importance to  develop error-correcting codes (ECCs) that can efficiently handle burst synchronization errors while maintaining low redundancy.

In the 1970s, Levenshtein laid the foundations to codes capable of correcting burst deletions in \cite{levenshtein1966binary,levenshtein1970asymptotically}, establishing a lower bound of $\log n+t-1$ bits of redundancy for codes correcting a single burst of $t$-deletions. In particular, asymptotically optimal binary codes with $\log n+1$ bits of redundancy were constructed for the case of $t=2$.\footnote{In this paper, unless stated otherwise, all logarithmic operations are base two.} Subsequently, several works extended these results to general $t\ge 2$ over binary and non-binary alphabets \cite{cheng2014codes,schoeny2017codes,saeki2018improvement}.
%In \cite{cheng2014codes}, one burst of $t\ge3$-deletion ECCs were presented for binary sequences with $t\log(n/t+1)$ bits of redundancy, which was improved to an optimal redundancy $\log n+(t-1)\log\log n+O(1)$ in \cite{schoeny2017codes}. Inspired by this work, Saeki \textit{et al.} \cite{saeki2018improvement} discussed one burst of $t$-deletion ECCs for $q$-ary sequences with $\log_q n+o(\log n)$ bits of redundancy.
As a further generalization, a number of studies \cite{gabrys2018codes,lenz2020optimal,schoeny2017codes,sima2020syndrome,song2023nonbinary,nguyen2024anew,song2024new,wang2024nonbinary} investigated codes capable of correcting one burst of at most $t$ deletions. Among these, the recent work of Song \textit{et al.} \cite{song2024new} achieved a redundancy of $\log n+8\log\log n+o(\log\log n)$ bits.
Moreover,
 Schoeny \textit{et al.} showed that correcting one burst of $t$-deletion ECCs is equivalent to correcting one burst of $t$-insertion ECCs in \cite{schoeny2017codes}. Building on this equivalence, Sun \textit{et al.} \cite{sun2025asymptotically} presented constructions  with $\log n+o(\log n)$ bits of redundancy for correcting one burst of $t$-edits over $q$-ary alphabets under certain conditions, where $q \geq 2$.

 Another notable contribution in \cite{schoeny2017codes} introduced burst ECCs capable of simultaneously handling deletions and insertions, constructing binary codes that correct one burst of $2$-deletion-$1$-insertion ($(2,1)$-DI) errors. This result was subsequently extended by Lu \textit{et al.} to one burst of $(t_1,1)$-DI errors \cite{lu2022tdeletion1}, and further to $(t_1,t_2)$-DI errors with $t_1 \geq 2t_2$ \cite{lu2023tdeletion}, establishing a lower bound of $\log(n-t_1+2)+t_1-1$ bits of redundancy.
  More recently, Sun \textit{et al.} \cite{sun2025asymptotically} provided explicit constructions of $q$-ary codes capable of correcting one burst of $(t_1,t_2)$-DI errors with $\log n+o(\log n)$ bits of redundancy, applicable to any non-negative integers $t_1,t_2$ and $q \geq 2$.

Beyond the single-burst setting, Sima \textit{et al.} \cite{sima2020syndrome} introduced a fundamentally different approach based on syndrome compression, enabling ECCs capable of correcting $m$ bursts of at most $t$-deletions with redundancy $4m\log n+o(\log n)$ bits. 
	Their work represents one of the first scalable frameworks that extend burst-deletion correction from a single-burst scenario to the general multi-burst regime while maintaining asymptotically optimal logarithmic redundancy. 
	In particular, focusing on the case of two bursts of exactly $t$-deletions, Ye \textit{et al.} \cite{ye2024codes} provided the first explicit construction with $5\log n+o(\log n)$ bits of redundancy. Shortly thereafter, Liu and Duman \cite{liu2025two} generalized the deletion-insertion framework by extending the single burst of $(t_1,t_2)$-DI model of \cite{sun2025asymptotically} to the two-burst setting for binary sequences, achieving redundancy of $11\log n+o(\log n)$.

Building upon these developments, this paper proposes binary ECCs capable of correcting one burst of $t$-deletions together with one burst of $t$-insertions. Throughout this work, we assume that two bursts in a sequence are non-overlapping. The proposed ECCs are developed by converting two burst errors into some erasures or substitutions, which is achieved by applying one burst of $t$-edit ECCs to appropriate substrings of the original sequence according to specific rules. This approach yields an explicit construction with redundancy on the order of $O(\log n)$. 
To the best of our knowledge, this is the first explicit construction for the channel model under consideration.

The remainder of the paper is organized as follows.
Section \ref{sec:Preliminaries} introduces some necessary preliminaries.
Section \ref{sec:1D1S} presents the explicit constructions of ECCs capable of correcting one burst of $t$-deletions together with one burst of $t$-substitutions.
Finally, Section \ref{sec:Conclusions} concludes this paper.

%As a key step, we first establish a necessary and sufficient condition for codes to be correctable under two bursts of $t$-edits. Building upon this characterization, we then design explicit constructions of burst ECCs by leveraging several related results in this literature.

\section{Preliminaries}\label{sec:Preliminaries}
We first establish the notation used throughout the paper.
For two non-negative integers $a$ and $b$ with $a < b$, we denote the ordered sets $\{a,a+1,\dots,b-1\}$ and $\{a,a+1,\dots,b\}$ by $[a,b)$ and $[a,b]$, respectively.
When $a=1$, we abbreviate $[1,b]$ as $[b]$.
The cardinality of a set $\cal{A}$ is denoted by $|\mathcal{A}|$.
For an $m\times n$ matrix $A$, we use $A[i]$ to denote its $i$-th row, where $i\in[m]$. %and  $A_{i,j}$ to denote the entry at the $i$-th row and $j$-th column, where $i\in[m],j\in[n]$.	
Let $\Sigma_q=\{0,1,\dots,q-1\}$, and denote by $\Sigma_q^n$ the set of all sequences with length $n$ over $\Sigma_q$.
For a $q$-ary sequence $\mathbf{x}=(x_1,x_2,\dots,x_n)\in\Sigma_q^n$, the length of $\mathbf{x}$ is denoted by $|\mathbf{x}|=n$.
Given an index set $\mathcal{I}=\{i_1,i_2,\dots,i_m\}\subseteq[n]$, the projection of $\mathbf{x}$ onto $\mathcal{I}$ is denoted by $\mathbf{x}_{\mathcal{I}}=(x_{i_1},x_{i_2},\dots,x_{i_m})$. We refer to $\mathbf{x}_{\mathcal{I}}$ as a \textit{substring} of $\mathbf{x}$ when $i_{j-1}=i_{j}-1$ for all $j\in[2,m]$.
If $x_i=x_{i+1}$ for every $i\in[i_1,i_m-1]$, the substring $(x_{i_1},x_{i_2},\dots,x_{i_m})$ is called a \textit{run} of $\mathbf{x}$. If $x_i=x_{i+2}$ for every $i\in[i_1,i_m-2]$, it forms an \textit{alternating substring} in $\mathbf{x}$. We remark that a run of $\mathbf{x}$ is also an alternating substring.
In addition, the number of runs in $\mathbf{x}$ and the $i$-th run of $\mathbf{x}$ are denoted by  $r(\mathbf{x})$ and $\mathbf{x}_i$, respectively, where $i \in [r(\mathbf{x})]$.

\begin{definition}\label{def:DS}
	For a binary sequence $\mathbf{x}=(x_1,x_2,\dots,x_n)\in\Sigma_2^n$ and  a positive integer $t$,
	if $(x_i,x_{i+1},\dots,x_{i+t-1})$ is deleted and $(x_j,x_{j+1},\dots,x_{j+t-1})$ is substituted by  $(x'_j,x'_{j+1},$ $\dots,x'_{j+t-1})$, where $i,j \in[n-t+1]$ with $i+t\le j$ or $j+t\le i$, $x'_k\in\{0,1\},k\in[j,j+t)$,  we say that a \textit{one burst of $t$-deletions together with one burst of $t$-substitutions} occur in $\mathbf{x}$.
\end{definition}

We denote by $\mathcal{M}_2(\mathbf{x})$ the set of sequences obtained from $\mathbf{x}$ by one deletion together with one substitution.
Let $\mathcal{B}_{1,1}^{DS}(\mathbf{x})$ and $\mathcal{B}^{E}(\mathbf{x})$
be the sets of sequences obtained from $\mathbf{x}$ by one burst of $t$-deletions together with one burst of $t$-substitutions and  one burst of $t$-edits, respectively.
%Furthermore, let $\mathcal{B}_{m}^D(\mathbf{x})$, $\mathcal{B}_{m}^S(\mathbf{x})$ and $\mathcal{B}_m^{E}(\mathbf{x})$ be sets of sequences obtained from $\mathbf{x}$ by $m$ bursts of $t$-deletions, $m$ bursts of $t$-substitutions, and $m$ bursts of $t$-edits, respectively. When $m=1$, we abbreviate as $\mathcal{B}^D(\mathbf{x})$, $\mathcal{B}^S(\mathbf{x})$ and $\mathcal{B}^{E}(\mathbf{x})$. Note that $m$ bursts of $t$-edits refer to a combination of $m_1$ bursts of $t$-insertions, $m_2$ bursts of $t$-deletions, and $m_3$ bursts of $t$-substitutions, where $m=m_1+m_2+m_3$.

\begin{definition}\label{def:DS1}
	A subset of $\Sigma_2^n$, denoted by $\mathcal{C}$, is referred to as a \textit{one burst of $t$-deletion one burst of $t$-substitution ECC} if for any two distinct sequences $\mathbf{x},\mathbf{y}\in\mathcal{C}$, it holds that
	\begin{align*}
		\mathcal{B}_{1,1}^{DS}(\mathbf{x})\cap \mathcal{B}_{1,1}^{DS}(\mathbf{y})=\varnothing.
	\end{align*}
\end{definition}

	\begin{definition}\label{def:bounded}
		Given a parameter $P$ (typically, of order
		$\log n$), a code $\mathcal{C}\subseteq\Sigma_2^n$ is referred to as a $P$-bounded one deletion one substitution ECC if, for any two distinct sequences $\mathbf{x}=\mathbf{u}\tilde{\mathbf{x}}\mathbf{v},\mathbf{y}=\mathbf{u}\tilde{\mathbf{y}}\mathbf{v}\in\mathcal{C}$ with $|\tilde{\mathbf{x}}|=|\tilde{\mathbf{y}}|\le P$, it holds that
		\begin{eqnarray*}
			\mathcal{M}_2(\tilde{\mathbf{x}})\cap \mathcal{M}_{2}(\tilde{\mathbf{y}})=\varnothing.
		\end{eqnarray*}
\end{definition}
%	\begin{remark}
	%		By convention, one burst of $t$-substitutions means that at most $t$ consecutive bits are substituted in a binary sequence. Throughout this paper, we restrict attention to the case where exactly $t$ consecutive bits are substituted, i.e., precisely $t$-complements occur.
	%	\end{remark}

In the following, we review several related ECC constructions, which will serve as the crucial components throughout this paper. Before that, we introduce a few useful definitions.
\begin{definition}
	For any binary sequence $\mathbf{x}=(x_1,x_2,\dots,x_n)\in\Sigma_2^n$,  the \textit{differential sequence} of $\mathbf{x}$ is defined as
	$
		\phi(\mathbf{x})\triangleq \big(\phi(x)_1,\phi(x)_2,\dots,\phi(x)_n\big)
	$,
	where $\phi(x)_i\triangleq x_{i}-x_{i-1} \pmod 2$ for $i\in[n]$ and we set $x_0=0$ for convenience.
\end{definition}

\begin{definition}
	Given an integer $\ell\ge 2$ and a constant $0<\epsilon<1/2$, a sequence $\mathbf{x}\in\Sigma_2^n$ is referred to as a \textit{strong $(\ell,\epsilon)$-locally balanced sequence} if for every substring of $\mathbf{x}$ with length $\ell'\ge\ell$, the number of ones lies within $[(1/2-\epsilon)\ell',(1/2+\epsilon)\ell']$.
\end{definition}
\begin{definition}
	A binary sequence $\mathbf{x}\in\Sigma_2^n$ is said to be \textit{$d$-regular}  if every substring of $\mathbf{x}$ with length at least $d\log n$ contains both $00$ and $11$.
\end{definition}

\begin{observation}\label{observation}
	If $\mathbf{x}\in\Sigma_2^n$ is a $d$-regular sequence, then  the length of each run and
	alternating substring in $\mathbf{x}$ is at most $d\log n$, since
	both $00$ and $11$ must appear in every substring of length $d\log n$.
\end{observation}

\begin{lemma}[Lem.15, \cite{sun2025codes}]\label{lem:regular}
	For two parameters $\ell=d\log n$ and $\epsilon$ being a constant less than one,	if both $\mathbf{x}$ and $\phi(\mathbf{x})$ are strong $(\ell,\epsilon)$-locally balanced
	sequences, then $\mathbf{x}$ is necessarily $d$-regular.
\end{lemma}

%	\subsection{Related Results}
In this paper,  unless otherwise specified, we set $\ell=1296\log n$ and $\epsilon={1/18}$. Denote $\mathcal{C}_{DS}^P\subseteq\Sigma_2^n$ as a binary $P$-bounded one deletion one substitution ECC, where $P\ge 6(\ell+3)$.

\begin{lemma}[Thm.4, \cite{sun2024binary}]\label{lem:mix}
	For $\mathbf{x}\in\Sigma_2^n$ with both $\mathbf{x}$ and $\phi(\mathbf{x})$ being strongly $(\ell,\epsilon)$-locally balanced, if $\mathbf{x}\in\mathcal{C}_{DS}^P$, there exists a function $g_{ds}:\Sigma_2^n\rightarrow \Sigma_2^{4\log n+12\log\log n+O(1)}$, such that given $g_{ds}(\mathbf{x})$ and $\mathbf{x}'\in\mathcal{M}_2(\mathbf{x})$,
	one can uniquely recover $\mathbf{x}$.
\end{lemma}

%\begin{remark}
%	In \cite{sun2025codes}, the original constructions of one burst of $t$-edit ECCs were developed with $\ell=625\log n$ and $\epsilon={1/15}$. By Lemma \ref{lem:le}, the results also hold for $\ell=1296\log n$ and $\epsilon={1/18}$. Hence, we adopt the latter parameters throughout this paper.
%\end{remark}

\begin{lemma}[Cor.1, \cite{song2022systematic}]\label{lem:ds}
	For any binary sequence $\mathbf{x}\in\Sigma_2^n$, there exists a function $g'_{ds}: \Sigma_2^n\rightarrow \Sigma_2^{6\log n+5}$, such that given $g'_{ds}(\mathbf{x})$ and $\mathbf{x}'\in\mathcal{M}_2(\mathbf{x})$,  one can uniquely recover $\mathbf{x}$.
\end{lemma}

Given a binary sequence $\mathbf{x}=(x_1,x_2,\dots,x_n)\in\Sigma_2^n$, assume that $t\mid n$, otherwise we append zeros at the very
end of each sequence such that its length is the smallest integer greater than $n$ and divisible by $t$.
Under this assumption, $\mathbf{x}$ can be represented by a $t\times (n/t)$ binary matrix as follows
\begin{align}\label{eqn:matrixt}
	\scalebox{0.9}{$ X_{t}=\left(\begin{array}{cccccc}
			X[1] \\
			X[2]\\
			\vdots\\
			X[t]
		\end{array}\right)=\left(\begin{array}{cccccc}
			x_1 & x_{t+1} &\cdots & x_{n-t+1}\\
			x_2 & x_{t+2} &\cdots & x_{n-t+2}\\
			\vdots & \vdots &\ddots &\vdots\\
			x_{t}& x_{2t} &\cdots&  x_{n}
		\end{array}\right)\in\Sigma_2^{ t\times {n\over t}},
		$}\end{align}
where $X[i],i\in[t],$ denotes the $i$-th row of $X_{t}$.

\begin{observation}\label{fact:location}
	If $\mathbf{x}\in\Sigma_2^n$ undergoes one burst of $t$-deletions and one burst of $t$-substitutions, then each row of $X_t$ experiences one deletion and at most one substitution, respectively. Moreover, the error positions in the $i$-th row may precede those in the first row by at most one position, where $i\in[2,t]$.
\end{observation}

Finally, we note that it is easy to construct an ECC capable of correcting one burst of $t$-edits in $\mathbf{x}$ by applying the binary single edit ECC proposed in \cite{levenshtein1966binary} with $\log n+1$ bits of redundancy to each row of $X_t$, yielding a total redundancy of $t(\log n + 1)$ for the resulting code. Consequently, the following lemma holds.
\begin{lemma}\label{lem:burst-edit}
	For any binary sequence $\mathbf{x}\in\Sigma_2^n$, there exists a function $g_b: \Sigma_2^n\rightarrow \Sigma_2^{t\log n+t}$, such that given $g_b(\mathbf{x})$ and $\mathbf{x}'\in\mathcal{B}^E(\mathbf{x})$, one can uniquely recover $\mathbf{x}$.
\end{lemma}
%\begin{remark}
%	More recently, one burst of $t$-edit ECCs with $\log n + o(\log n)$ bits of redundancy were introduced in \cite{sun2025codes}, requiring that $\mathbf{x}$ is a codeword of a $P$-bounded one burst of $t$-deletion ECC and $(X[1],X[2],\dots,X[t])$ forms a strong $(\ell,\epsilon)$-locally balanced sequence, where $X[i]$ denotes the $i$-th row of $X_t$. However, these codes are not effective for the proposed framework in this paper, as they incur significantly higher redundancy.
%\end{remark}

\section{One Burst of $t$-Deletion One Burst of $t$-Substitution ECCs}\label{sec:1D1S}
In this section, ECCs capable of correcting one burst of $t$-deletions together with one burst of $t$-substitutions are developed by means of a matrix form of the original sequence.% The approximate positions of the two bursts are then identified by applying an ECC capable of correcting one deletion together with one substitution to the first row of the matrix. Subsequently, by employing a one burst of $t$-edit ECC to each substring of the original sequence according to certain designed rules, the targeted errors can be transformed to some erasures or substitutions, which are then corrected by a Generalized Reed-Solomon (GRS) code.

%We highlight that one burst of $t$-substitutions considered in this paper is specifically refers to a burst of complement errors, rather than at most $t$-substitutions as considered in the channel models of \cite{sun2024binary,sun2025codes}. %The resulted  codes are obtained  following the $1$-deletion $1$-substitution ECCs introduced in Lemma \ref{lem:mix}.

As shown in Equation \eqref{eqn:matrixt} and Observation \ref{fact:location},  representing  any binary sequence $\mathbf{x}\in\Sigma_2^n$ in its matrix form $X_t$ reveals that one burst of $t$-deletions together with one burst of $t$-substitutions in $\mathbf{x}$ cause one deletion and at most one substitution in $X[1]$.
In the sequel, we illustrate two different cases following the error types in $X[1]$.

In the first case, we assume that one deletion together with one substitution occur within $X[1]$. Under this assumption, the following lemma can be established. %Note that it is reasonable to assume that  the deletion always occurs after the substitution from the Lemma 11 of \cite{sun2024binary}.

\begin{lemma}\label{lem:1d1s}
	Let $\mathbf{x}\in\Sigma_2^n$ be a $d$-regular sequence. If $\mathbf{x}'$ is obtained from $\mathbf{x}$ by  one deletion together with one substitution, then one of the following conclusions holds:
	\begin{itemize}
		\item[a)] There exist two distinct runs $\mathbf{x}_i$ and $\mathbf{x}_j$ with $i,j\in[r(\mathbf{x})]$,  such that the substitution occurs in $\mathbf{x}_i$ and the deletion occurs in $\mathbf{x}_j$, respectively.
		\item[b)] There exists a substring $\mathbf{x}_{\mathcal{I}}$ with $|\mathcal{I}|\le 3d\log n$, such that the one deletion and one substitution occur in $\mathbf{x}_{\mathcal{I}}$.
	\end{itemize}
\end{lemma}

\begin{proof}
	Assume that $\mathbf{x}'$ is obtained from $\mathbf{x}$ by  one substitution and one deletion at the $i_1$-th and $i_2$-th components, respectively.
	If there exist two indices $j_1,j_2\in[n]$ such that $\mathbf{x}'$ can also be obtained from $\mathbf{x}$ by  one substitution and one deletion at the $j_1$-th and $j_2$-th components, respectively, then, by symmetry, we consider the following cases.
	
	(i) $i_1\le j_1<i_2\le j_2$. In this case, we have
	\begin{align*}
		\mathbf{x}'=&x_1\cdots x_{i_1-1}\overline{x}_{i_1}x_{i_1+1}\cdots x_{j_1-1}x_{j_1}x_{j_1+1}\cdots x_{i_2-1}x_{i_2+1}\\
		&x_{i_2+2}\cdots ~~x_{j_2}~x_{j_2+1}\cdots x_n \\
		=&x_1\cdots x_{i_1-1}x_{i_1}x_{i_1+1}\cdots x_{j_1-1}\overline{x}_{j_1}x_{j_1+1}\cdots x_{i_2-1}~~x_{i_2}~\\
		&x_{i_2+1}\cdots x_{j_2-1}x_{j_2+1}\cdots x_n,
	\end{align*}
	where $\overline{x}_{i}$, $i\in\{i_1,j_1\}$ represents the substitution error in $\mathbf{x}$.
	Thus, by comparing the symbols of the corresponding positions, we obtain that when  $i_1\ne j_1$,
	\begin{align}\label{eqn:i1j1}
	x_{i_1}=\overline{x}_{i_1},x_{j_1}=\overline{x}_{j_1},
	\end{align}
	and
	\begin{align}\label{eqn:i2=j2}
		x_{i_2}=x_{i_2+1}=\cdots= x_{j_2-1}=x_{j_2}.
	\end{align}
Note that the identity \eqref{eqn:i1j1} is impossible, which means $i_1=j_1$, i.e., the substitutions in $\mathbf{x}$ should occur at the same component.
	By \eqref{eqn:i2=j2}, $\mathbf{x}_{[i_2,j_2]}$ is a run of $\mathbf{x}$.
	Denoting $\mathbf{x}_i=\mathbf{x}_{[i_1]}$ and $\mathbf{x}_j=\mathbf{x}_{[i_2,j_2]}$,
	Conclusion a) holds.

	(ii) $i_2< i_1\le j_1< j_2$. In this case, we have		
	\begin{align*}
		\mathbf{x}'=&x_1\cdots x_{i_2-1}x_{i_2+1}\cdots  x_{i_1-1}~~\overline{x}_{i_1}~x_{i_1+1}\cdots ~~x_{j_1}~x_{j_1+1}\\
		&x_{j_1+2}\cdots ~~x_{j_2}~x_{j_2+1}\cdots x_n\\
		=&x_1\cdots x_{i_2-1}~~x_{i_2}~\cdots x_{i_1-2}x_{i_1-1}~~x_{i_1}~\cdots x_{j_1-1}~~\overline{x}_{j_1}~\\
		&x_{j_1+1}\cdots x_{j_2-1}x_{j_2+1}\cdots  x_n.
	\end{align*}

Since $i_1\le j_1$, we can obtain 
$
x_{i_2}=x_{i_2+1}=\cdots= x_{i_1-1}=\overline{x}_{i_1}$,
$x_{i_1}=x_{i_1+1}=\cdots= x_{j_1-1}=x_{j_1}$, and
$\overline{x}_{j_1}=x_{j_1+1}=\cdots=x_{j_2-1}=x_{j_2}$.	
%\end{align*}

	Thus, $\mathbf{x}_{[i_2,i_1-1]}$, $\mathbf{x}_{[i_1,j_1]}$ and $x_{[j_1+1,j_2]}$ are three runs of $\mathbf{x}$.
	Denoting $\mathcal{I}=[i_2,j_2]$, we know $|\mathcal{I}|\le 3d\log n$  with Observation \ref{observation}.
	Consequently, Conclusion b) holds.

	(iii)
	$i_1<i_2\le j_1<j_2$. In this case, we have
	\begin{align*}
		\mathbf{x}'=&x_1\cdots x_{i_1-1}\overline{x}_{i_1}x_{i_1+1}\cdots  x_{i_2-1}x_{i_2+1}\cdots ~~x_{j_1}~x_{j_1+1}\cdots\\
		& ~x_{j_2}~x_{j_2+1}\cdots x_n\\
		=&x_1\cdots x_{i_1-1}x_{i_1}x_{i_1+1}\cdots x_{i_2-1}~~x_{i_2}~\cdots x_{j_1-1}~~\overline{x}_{j_1}~\cdots\\ &x_{j_2-1}x_{j_2+1}\cdots x_n.
	\end{align*}
	Thus, we can obtain $\overline{x}_{i_1}=x_{i_1}$, and
	$
		x_{i_2}=x_{i_2+1}=\cdots= x_{j_1-1}=x_{j_1}$ and
		$\overline{x}_{j_1}=x_{j_1+1}=\dots=x_{j_2-1}=x_{j_2}$.
	%\end{align*}
	It is clear that $\overline{x}_{i_1}=x_{i_1}$ is not admissible.
	
	Using a similar line of argument, by comparing the symbols of the corresponding positions,  we can prove Conclusion a) holds for the cases that $i_1\le  j_1<j_2\le  i_2$, 	$i_2\le j_2< i_1\le j_1$ and $i_2\le j_2< j_1\le i_1$, while Conclusion b) holds for the case of $i_2\le j_1<i_1\le j_2$. Moreover, the remaining configurations $i_1\le j_2< i_2\le j_1$, $i_1\le j_2< j_1\le i_2$, $i_1<i_2\le j_2<j_1$,	$i_2\le  j_1< j_2\le  i_1$ and	$i_2<i_1\le j_2<j_1$ are not admissible under the considered error model. This completes the proof.
\end{proof}

In the second case, we consider the scenario in which no substitution occurs in $X[1]$, and hence $X[1]$ only experiences a single deletion. The following lemma then holds.

\begin{lemma}\label{lem:1d}
	For any binary sequence $\mathbf{x}\in\Sigma_2^n$, if $\mathbf{x}'$ is obtained from $\mathbf{x}$ by  one deletion, then one can find a unique run $\mathbf{x}_i$,  such that the deletion occurs in $\mathbf{x}_{i},i\in[r(\mathbf{x})]$.
\end{lemma}

\begin{proof}
	Assume that $\mathbf{x}'$ is obtained from $\mathbf{x}$ by one deletion at the $i_1$-th component, where $i_1\in[n]$.
	If there exists another index $i_2\in[n]$ such that $\mathbf{x}'$ can also be obtained from $\mathbf{x}$ by one deletion at the $i_2$-th component, without loss of generality, letting $i_1< i_2$, we have
	\begin{align*}
		\mathbf{x}'=&x_1\cdots x_{i_1-1}x_{i_1+1}\cdots ~~x_{i_2}~x_{i_2+1}\cdots x_n \\
		=&x_1\cdots x_{i_1-1}~~x_{i_1}~\cdots x_{i_2-1}x_{i_2+1}\cdots x_n.
	\end{align*}
	Thus, by comparing the symbols of the corresponding positions, we obtain
	$x_{i_1}=x_{i_1+1}=\cdots= x_{i_2-1}=x_{i_2}$.
	Denoting $\mathbf{x}_i=\mathbf{x}_{[i_1,i_2]}$, the proof is completed.
\end{proof}

Denote the $j$-th run of $X[1]$ as
\begin{equation}\label{eqn:def_pj}
X[1]_{j}=X[1]_{[p_{j-1}+1,p_j]},j\in[r(X[1])],
\end{equation}
and some intervals as
\begin{align}\label{eqn:I_i}
	\mathcal{I}_j=[p_{j-1}t+1,p_{j}t],\quad j\in[r(X[1])].
\end{align}
For convenience, we assign $p_0=0$ and $p_{r(X[1])}=n/t$.

Theorem \ref{thm:1D1S} is a direct consequence of Lemmas \ref{lem:1d1s} and \ref{lem:1d}.
\begin{theorem}\label{thm:1D1S}
	For a given binary sequence $\mathbf{x}\in\Sigma_2^n$, let $X_t$ be the matrix form of $\mathbf{x}$ following \eqref{eqn:matrixt} and $g_{ds}(\mathbf{x})$ be the function defined in Lemma \ref{lem:mix}.
	Supposing that both $X[1]\in\mathcal{C}_{DS}^P$ and $\phi(X[1])$ are strong $(\ell,\epsilon)$-locally balanced sequences, if $\mathbf{x}'$ is obtained from  $\mathbf{x}$ by one burst of $t$-deletions and one burst of $t$-substitutions, one of the following conclusions holds:
	\begin{itemize}
		\item [a)] If there exist one deletion together with one substitution in $X[1]$, one can find two intervals $\mathcal{J}_1\subseteq\mathcal{I}_{i_1-1}\cup \mathcal{I}_{i_1}$, $\mathcal{J}_2\subseteq\mathcal{I}_{i_2-1}\cup \mathcal{I}_{i_2}$ such that
		one burst of $t$-deletions occurs at $\mathbf{x}_{\mathcal{J}_1}$ and one burst of $t$-substitutions occurs at $\mathbf{x}_{\mathcal{J}_2}$, respectively.
		\item[b)] If there exists only one deletion in $X[1]$, one can find an interval $\mathcal{J}_1\subseteq\mathcal{I}_{i_1-1}\cup \mathcal{I}_{i_1}$ such that
		one burst of $t$-deletions occurs at $\mathbf{x}_{\mathcal{J}_1}$.
		\item [c)] One can find an interval $\mathcal{J}\subseteq[n]$ with $|\mathcal{J}|\le 3dt\log(n/ t)+t-1$ such that the one  burst of $t$-deletions and one burst of $t$-substitutions occur at $\mathbf{x}_{\mathcal{J}}$.
	\end{itemize}
\end{theorem}
\begin{proof}
Assume that $X'_t$ is the matrix form of $\mathbf{x}'$. By Observation \ref{fact:location}, if $\mathbf{x}'\in\mathcal{B}_{1,1}^{DS}(\mathbf{x})$, then it follows that $X'[1]\in\mathcal{M}_{2}(X[1])$.
	According to Lemma \ref{lem:mix}, the sequence $X[1]$ can be uniquely recovered from $g_{ds}(X[1])$ and $X'[1]$  since both $X[1]\in\mathcal{C}_{DS}^P$ and $\phi(X[1])$ are strong $(\ell,\epsilon)$-locally balanced sequences.
Furthermore, by Lemma \ref{lem:regular}, it is clear that $X[1]$ is a $d$-regular sequence of length $n/t$. We distinguish three different cases based on Lemma \ref{lem:1d1s} and Lemma \ref{lem:1d} in the following.

	{\bf Case I:} Lemma \ref{lem:1d} holds. Assume that one deletion occurs at the $i_1$-th run of $X[1]$.
	Note that if the deletion occurs at the $i$-th symbol of $X[1]$, then the corresponding burst should start
at the $j$-th symbol with $j\in [(i-2)t+2,(i-1)t+1]$.
Thus, by \eqref{eqn:matrixt}, \eqref{eqn:def_pj}, and \eqref{eqn:I_i}, we have that the burst of $t$-deletions occurs at
$\mathbf{x}_{\mathcal{J}_1}\subseteq \mathbf{x}_{\mathcal{I}_{i_1-1}\cup \mathcal{I}_{i_1}}$,
	where $\mathcal{J}_1=[(p_{i_1-1}-1)t+2,p_{i_1}t]$. Then Conclusion b) follows.
	
    {\bf Case II:}	Lemma \ref{lem:1d1s}-a) holds. In this case, assume the deletion and substitution occur at the $i_1$-th and $i_2$-th run of $X[1]$, respectively.
	Similar to Case I, Conclusion a) follows by \eqref{eqn:def_pj} and \eqref{eqn:I_i}, i.e., the burst of $t$-deletions in $\mathbf{x}$ occurs at
	$
		\mathbf{x}_{\mathcal{J}_1}\subseteq \mathbf{x}_{\mathcal{I}_{i_1-1}\cup \mathcal{I}_{i_1}},
	$
	and the burst of $t$-substitutions in $\mathbf{x}$ occurs at
	$
		\mathbf{x}_{\mathcal{J}_2}\subseteq \mathbf{x}_{\mathcal{I}_{i_2-1}\cup \mathcal{I}_{i_2}},
	$
	where $
		\mathcal{J}_1={[(p_{i_1-1}-1)t+2,p_{i_1}t]}$,
		$\mathcal{J}_2=[(p_{i_2-1}-1)t+2,p_{i_2}t]$.

	{\bf Case III:} Lemma \ref{lem:1d1s}-b) holds. Let 
	$$\mathcal{I}=[i,i']\subseteq\left[{n\over t}\right],|\mathcal{I}|\le 3d\log \left({n\over t}\right).$$
	Then the burst of $t$-deletions and burst of $t$-substitutions should occur in $\mathbf{x}_{\mathcal{J}}$ according to \eqref{eqn:matrixt}, \eqref{eqn:def_pj}, and \eqref{eqn:I_i}, where
	$$\mathcal{J}=[(i-2)t+2,i't].$$
	It is clear that
	$$|\mathcal{J}|\le 3dt\log \left({n\over t}\right)+t-1,$$	
i.e., Conclusion c) follows, which completes the proof.
\end{proof}

In what follows, we describe the explicit constructions of ECCs for correcting one burst of $t$-deletions and one burst of $t$-substitutions.
	For any binary sequence $\mathbf{x}\in\Sigma_2^n$ with both $X[1]\in\mathcal{C}_{DS}^P$ and $\phi(X[1])$ being strong $(\ell,\epsilon)$-locally balanced sequences, applying a one burst of $t$-edit ECC $g_b(\cdot)$ presented in Lemma \ref{lem:burst-edit} to each substring $\mathbf{x}_{\mathcal{I}_j}$ of $\mathbf{x}$,  we can obtain
\begin{align*}
	\bar{g}_b(\mathbf{x})\triangleq\left(g_b(\mathbf{x}_{\mathcal{I}_1}),g_b(\mathbf{x}_{\mathcal{I}_2}),\dots,g_b(\mathbf{x}_{\mathcal{I}_{r(X[1])}})\right),
\end{align*}
where $\mathcal{I}_j,j\in[r(X[1])],$ is defined in \eqref{eqn:I_i}.
Since both $X[1]$ and $\phi(X[1])$ are two strong $(\ell,\epsilon)$-locally balanced sequences, from Lemma \ref{lem:regular} and Observation \ref{observation}, we have  $|\mathcal{I}_i|\le dt\log (n/t),i\in[r(X[1])]$.
By Lemma \ref{lem:burst-edit}, viewing $g_b(\cdot)$ as a binary representation of an integer,  we can write
\begin{align*}
	|g_b(\cdot)| &\le t\log (dt\log (n/t))+t\nonumber\\\
	&=t\log\log (n/t)+t\log dt+t.
\end{align*}
In other words, the length of $g_b(\cdot)$ is $O(\log\log n)$.
%Based on this, the explicit construction of one burst of $t$-deletion one burst of $t$-substitution ECC can be developed as follows.

Let $\rho=3d\log({n/ t})+1$ be a positive integer, and $\mathcal{L}_j,j\in[\lceil n/\rho\rceil-1],$ be intervals of length $2\rho$, where
\begin{align}\label{eqn:interval}
	\hspace{-0.1cm}\mathcal{L}_j=\left\{\begin{array}{ll}
		[(j-1)\rho+1,(j+1)\rho], & \hspace{-0.3cm}\text{for~~} j\in[\lceil n/\rho\rceil -2],\\
		{[(j-1)\rho+1,n]},&\hspace{-0.3cm} \text{for~~} j=\lceil n/\rho\rceil -1.
	\end{array}\right.
\end{align}

\begin{construction}\label{con3}
	For any $\mathbf{x}\in\Sigma_2^n$ with both $X[1]\in\mathcal{C}_{DS}^P$ and $\phi(X[1])$ being strong $(\ell,\epsilon)$-locally balanced sequences,
	define
	\begin{align}\label{eqn:bar-h}
		h^{(a)}_i(\mathbf{x})\triangleq\sum_{j\in[\lceil n/\rho\rceil-1]\atop j\equiv a \bmod 2} g'_{ds}\big(X[i]_{\mathcal{L}_j}\big) \pmod{N},
	\end{align}
	where $a\in\{0,1\},i\in[2,t]$ and $N=2^{6\log (2\rho)+5}$.
	
	Then, let
	\begin{align*}
		\begin{aligned}	&f_{ds}(\mathbf{x})=\\
				&\left\{g_{ds}\big(X[1]\big),(\bar{g}_b(\mathbf{x}),\bm 0),h^{(a)}_i(\mathbf{x}),
				a\in\{0,1\},i\in[2,t]\right\},
			\end{aligned}\end{align*}
where $(\bar{g}_b(\mathbf{x}),\bm 0)\in \mathbb{F}^n_q$ with $q\geq n$ being the smallest prime greater than $n$.	
	For some fixed parameters $g\in[0,2^{4\log n+12\log\log n+O(1)}),\delta\in \mathbb{F}^n_q$ and $h_{i,a}\in[0,N), a\in\{0,1\},i\in[2,t]$, a one burst of $t$-deletion one burst of $t$-substitution ECC can be defined by
	\begin{align*}
		\mathcal{C}^{\delta}_{ds}\triangleq\Big\{\mathbf{x}\in\Sigma_2^n: &~\text{Both $X[1]\in\mathcal{C}_{DS}^P$ and $\phi(X[1])$ are strong}\\
		&\text{$(\ell,\epsilon)$-locally balanced sequences},\\
		&f_{ds}(\mathbf{x})=(g,\mb{r},h_{i,a}),
		a\in\{0,1\},i\in[2,t],\\
 &\text{ with }\mb{r}\in \mathcal{C}_0+\delta\Big\},
	\end{align*}
where $\mathcal{C}_0$ is an $[n,n-6,d_{\mathrm{min}}=7]_q$ Reed-Solomon (RS) code and $\mathcal{C}_0+\delta\triangleq\{\mathbf{c}+\delta:\mb{c}\in\mathcal{C}_0\}$.
\end{construction}

\begin{theorem}\label{thm:Tt}	
	The code $\mathcal{C}_{ds}$ generated by Construction  \ref{con3} can correct one  burst of $t$-deletions and one burst of $t$-substitutions with at most $10\log n+12t\log\log n+O(1)$
	bits of redundancy.
\end{theorem}
\begin{proof}	
	By Construction  \ref{con3}, both $X[1]\in\mathcal{C}_{DS}^P$ and $\phi(X[1])$ are strong $(\ell,\epsilon)$-locally balanced sequences.
We distinguish three different cases according to Theorem \ref{thm:1D1S}.
	
	{\bf Case I:} Theorem \ref{thm:1D1S}-a) holds, i.e, the approximate positions of the burst-deletion and burst-substitution are located in  $\mathbf{x}_{\mathcal{I}_{i_1-1}\cup\mathcal{I}_{i_1}\cup\mathcal{I}_{i_2-1}\cup\mathcal{I}_{i_2}}$.
	In this case, for any $i\in[r(X[1])]\backslash\{i_1-1,i_1,i_2-1,i_2\}$, the substring $\mathbf{x}_{I_{i}}$ and $g_b(\mathbf{x}_{I_i})$ are known.
	Thus, $\bar{g}_b(\mathbf{x})$ experiences two bursts of  two erasures, which can be recovered since $\mb{r}=(\bar{g}_b(\mathbf{x}),\mathbf{0})\in\mathcal{C}_0+\delta$ and  $d_{\mathrm{min}}(\mathcal{C}_0)=7$.
	
	{\bf Case II:} Theorem \ref{thm:1D1S}-b) holds.
 In this case,  $g_b(\mathbf{x}_{I_i}),i\in\{i_1-1,i_1\}$ is unknown. In addition, at most two substitutions may occur at
 $\bar{g}_b(\mathbf{x})$ owing to one burst of $t$-substitutions in $\mathbf{x}$, resulting in two erasures and two substitutions
 in $\bar{g}_b(\mathbf{x})$, which can be recovered since $\mb{r}=(\bar{g}_b(\mathbf{x}),\mathbf{0})\in\mathcal{C}_0+\delta$ and  $d_{\mathrm{min}}(\mathcal{C}_0)=7$.

For these two cases, after recovering $g_b(\mathbf{x}_{\mathcal{I}_i}),i\in r(X[1])$ from $\bar{g}_b(\mathbf{x}')$, the precise locations of four erasures or two erasures together with two substitutions in $\bar{g}_b(\mathbf{x})$ can be identified, denoting by $g_b(\mathbf{x}_{\mathcal{I}_{i_1-1}}), g_b(\mathbf{x}_{\mathcal{I}_{i_1}}),g_b(\mathbf{x}_{\mathcal{I}_{i_2-1}})$ and $g_b(\mathbf{x}_{\mathcal{I}_{i_2}})$.
	Correspondingly, the approximate position of the two bursts in $\mathbf{x}$ can be located, namely, $\mathbf{x}_{\mathcal{I}_{i_1-1}\cup \mathcal{I}_{i_1}\cup\mathcal{I}_{i_2-1}\cup \mathcal{I}_{i_2}}$.
	To recover the original sequence $\mathbf{x}$,
	with $g_b(\mathbf{x}_{\mathcal{I}_i})$ being an ECC capable of correcting one burst of $t$-edits in $\mathbf{x}_{\mathcal{I}_i}$,
	it is sufficient to find four intervals $\mathcal{K}_i$ such that $\mathbf{x}'_{\mathcal{K}_i}\in \mathcal{B}^E(\mathbf{x}_{\mathcal{I}_{i}})$ for $i\in\{i_1-1,i_1,i_2-1,i_2\}$.
	Without loss of generality, let $i_1-1>i_2$ below.

	Following Theorem \ref{thm:1D1S}-a) and \ref{thm:1D1S}-b), for the sequence $\mathbf{x}'=(x'_1,x'_2,\dots,x'_{n-t})\in\Sigma_2^{n-t}$ obtained from $\mathbf{x}=(x_1,x_2,\dots,$ $x_n)\in\Sigma_2^n$,
	it holds that
	\begin{align*}
		x'_i=\left\{\begin{array}{ll}
			x_i, & \text{~if~} i\in[(p_{i_2-1}-1)t+1] \\
			&\hspace{1cm}\cup [p_{i_2}t+1,(p_{i_1-1}-1)t+1],\\
			x_i/\bar{x}_i, &  \text{~if~} i\in[(p_{i_2-1}-1)t+2,p_{i_2}t],\\
			x_{i-t}, &  \text{~if~} i\in[p_{i_1}t+1,n],
		\end{array}\right.
	\end{align*}
	where $\bar{x}_i$ indicates that a substitution error occur at the $i$-th component of $\mathbf{x}$.
	As a result, we can obtain
	\begin{align}
		\mathbf{x}'_{[p_{i_2-2}t+1,p_{i_2-1}t]}\in\mathcal{B}^S(\mathbf{x}_{\mathcal{I}_{i_2-1}}),\label{eqn:sub}\\
		\mathbf{x}'_{[p_{i_2-1}t+1,p_{i_2}t]}\in\mathcal{B}^S(\mathbf{x}_{\mathcal{I}_{i_2}}),\label{eqn:sub1}
	\end{align}
	and		
	$
		\mathbf{x}'_{[p_{i_1-2}t+1,p_{i_1-1}t-t]}\in\mathcal{B}^D(\mathbf{x}_{\mathcal{I}_{i_1-1}})$,
		$\mathbf{x}'_{[p_{i_1-1}t+1,p_{i_1}t-t]}\in\mathcal{B}^D(\mathbf{x}_{\mathcal{I}_{i_1}}).
$
	
	It is remarkable that the formulas \eqref{eqn:sub} and \eqref{eqn:sub1} hold since one burst of $t$-substitutions in $\mathbf{x}$ means that at most $t$ symbols, rather than exactly $t$ consecutive symbols in $\mathbf{x}$ are substituted.
	Thus, the four intervals exist by denoting $\mathcal{K}_{i_2-1}=[p_{i_2-2}t+1,p_{i_2-1}t]$, $\mathcal{K}_{i_2}=[p_{i_2-1}t+1,p_{i_2}t]$,
	$\mathcal{K}_{i_1-1}=[p_{i_1-2}t+1,p_{i_1-1}t-t]$, and
	$\mathcal{K}_{i_1}=[p_{i_1-1}t+1,p_{i_1}t-t]$.

	{\bf Case III:} Theorem \ref{thm:1D1S}-c) holds. Note that
   this case arises only when Lemma \ref{lem:1d1s}-b) holds.
	 Thus, from $g_{ds}\big(X[1]\big)$, we infer one deletion and one substitution are confined to a substring of $X[1]$ with length at most $3d\log(n/t)$. Following Observation \ref{fact:location}, one deletion and one substitution occur in $X[i]_{\mathcal{I}},i\in[2,t]$ with $|\mathcal{I}|\le 3d\log(n/t)+1$.
	 By \eqref{eqn:interval}, there exists an index $j_1$ such that $\mathcal{I}\subseteq \mathcal{L}_{j_1}$.
	
	Assuming that $\mathcal{L}_{j_1}=[\lambda_0,\lambda_1]$, thus, we know for $i\in[2,t]$,
	\begin{equation}\label{eqn:lambda2}
		\begin{aligned}
			X[i]_{[1,\lambda_0)}&=X'[i]_{[1,\lambda_0)}, \\
			X[i]_{[\lambda_1+1,n]}&=X'[i]_{[\lambda_1,n-1]},
		\end{aligned}
	\end{equation}
	and
	\begin{align}\label{eqn:x'x2}
		X'[i]_{[\lambda_0,\lambda_1-1]}\in\mathcal{M}_{2}(X[i]_{\mathcal{L}_{j_1}}).
	\end{align}
	
		According to \eqref{eqn:bar-h}, while $j_1\bmod 2=a$,	we can compute
	\begin{align}\label{eqn:g'ds}
		g'_{ds}(X[i]_{\mathcal{L}_{j_1}})=h^{(a)}_i(\mathbf{x})-\sum_{j\in[\lceil n/\rho\rceil-1]\backslash\{j_1\}\atop j \bmod 2=a} g'_{ds}\big(X[i]_{\mathcal{L}_j}\big) \notag\\
		\pmod{N},
	\end{align}
where $a\in\{0,1\}$ and $N=2^{6\log (2\rho)+5}$.

	From Lemma \ref{lem:ds}, it holds that
	$g'_{ds}(X[i]_{\mathcal{L}_j})< N$, where  $i\in[2,t]$, $j\in[\lceil n/\rho\rceil-1].
	$ Therefore, $X[i]_{\mathcal{L}_{j_1}}$ can be recovered by $X'[i]_{[\lambda_0,\lambda_1-1]}$ and \eqref{eqn:g'ds}.
	With \eqref{eqn:lambda2} and \eqref{eqn:x'x2}, $X[i]$ can be uniquely recovered.
Thus, we obtain the original sequence $\mathbf{x}$.
	
Observe that $\mathcal{C}_0$ is an $(n-6)$-dimensional subspace of 
	$\mathbb{F}_{q}^{n}$.
	Hence, the full space $\mathbb{F}_{q}^{n}$ can be decomposed into ${q}^{6}$ disjoint cosets of $\mathcal{C}_0$. In particular, there exist vectors $\delta_i \in \mathbb{F}_{q}^{n}$ such that
	$\mathcal{C}_0^i \triangleq \mathcal{C}_{1} + \delta_i,i \in [{q}^{6}]$,
	and the collection $\{\mathcal{C}_0^i:i\in[{q}^{6}]\}$ forms a partition of $\mathbb{F}_{q}^{n}$.
	For each $i\in[{q}^{6}]$, we can accordingly construct a code $\mathcal{C}_{ds}^{\delta_i}$.
	Since each coset is generated by adding a fixed vector to $\mathcal{C}_0$, all cosets inherit the same minimum distance, namely,
	%\begin{align*}
	$d_{\min}(\mathcal{C}_0^i) = d_{\min}(\mathcal{C}_0)$.
	%\end{align*}
	Therefore, every coset $\mathcal{C}_0^i$ admits identical error-correcting capability. 
	
	Define $\mathcal{C}^* \triangleq \bigcup_{i\in[{q}^{6}]} \mathcal{C}_{ds}^{\delta_i}$.
	By the pigeonhole principle, there exists at least one index $i \in [{q}^{6}]$ such that $|\mathcal{C}_{ds}^{\delta_i}| \ge |\mathcal{C}^*|/{q}^{6}$.
	Consequently, one can select a coset representative $\delta=\delta_i$ for which the induced code incurs no more than $6\log n+O(1)$ bits of redundancy.

	According to Equation \eqref{eqn:bar-h}, Lemmas \ref{lem:mix} and \ref{lem:ds}, the length of $f_{ds}(\mathbf{x})$  (viewed as a binary string) satisfies
	\begin{align*}
		|f_{ds}(\mathbf{x})|=&|g_{ds}(X[1])|+\sum_{i=2}^t\sum_{a=0}^1\left |h^{(a)}_i(\mathbf{x})\right|\\
		=&4\log\left({n\over t}\right)+12\log\log \left({n\over t}\right)+O(1)\\
		&+2(t-1)\log N+ 6\log n +O(1)\\
		\le &10\log{n}+12t\log\log n +O(1),
	\end{align*}
	where the last equation holds due to $\log N=6\log 2\rho+5$ and $\rho=O(\log n)$, 
	 completing the proof.
\end{proof}

\begin{remark}
The	syndrome compression technique proposed in \cite{sima2020syndrome} can be directly applied to construct ECCs capable of correcting one burst of $t$-deletions and one burst of $t$-substitutions with $8\log n+o(\log n)$ bits of redundancy.  However, this approach does not yield an explicit construction, whereas the codes proposed in this paper are explicit.
\end{remark}

\section{Conclusions}\label{sec:Conclusions}
In this paper, we presented an explicit construction of binary ECCs with $10\log n+12t\log\log n+O(1)$ bits of redundancy, capable of correcting one burst of $t$-deletions together with one burst of $t$-substitutions. To achieve this, we represent the original sequence in a matrix form and convert two burst errors to some erasures or substitutions, which can be corrected by a coset of an RS code.   %Particularly, when one burst of $t$-substitutions is replaced to one burst of $t$-complements, the redundancies of those codes can be further reduced to $5\log n+o(\log n)$ bits.

	\ifCLASSOPTIONcaptionsoff
	\newpage
	\fi

	\bibliographystyle{IEEEtran}
	\bibliography{myreference}

\end{document}